\documentclass[11pt]{article}
\usepackage[top=1 in,bottom=1in, left=1 in, right=1 in]{geometry}
\usepackage{hyperref}                                                          % paper

\usepackage{amsmath} % assumes amsmath package installed
\usepackage{amssymb}  % assumes amsmath package installed
\usepackage[T1]{fontenc}
\usepackage{bbm}
\usepackage{color}
\usepackage{amsthm}
\usepackage{dsfont}
\usepackage{soul}

\DeclareMathOperator{\supp}{supp}
\DeclareMathOperator{\Span}{span}
\DeclareMathOperator{\tr}{tr}
\DeclareMathOperator{\diag}{diag}

\newtheorem{theorem}{Theorem}[section]
\newtheorem{corollary}{Corollary}[section]
\newtheorem{lemma}{Lemma}[section]
\newtheorem{proposition}{Proposition}[section]

\theoremstyle{definition}
\newtheorem{definition}{Definition}[section]
\newtheorem{remark}{Remark}[section]
\newtheorem{example}{Example}[section]

\title{\LARGE \bf
Ergodic Properties of Quantum Markov Semigroups
}

\author{Nicolas Mousset and Nina H. Amini% <-this % stops a space
\thanks{This work was supported by the ANR projects QuanTEdu-France (ANR-22-CMAS-0001), Q-COAST (ANR-19-CE48-0003) and IGNITION (ANR-21- CE47-0015).}% <-this % stops a space
\thanks{N. Mousset is with CNRS, Laboratoire des Signaux et Syst\`{e}mes (L2S), Universit\'{e} Paris-Saclay,
        CentraleSupélec, 91190 Gif-sur-Yvette, France
        {\tt\small nicolas.mousset@centralesupelec.fr}}%
\thanks{N. H. Amini is with CNRS, Laboratoire des Signaux et Syst\`{e}mes (L2S), Universit\'{e} Paris-Saclay,
        CentraleSupélec, 91190 Gif-sur-Yvette, France
        {\tt\small nina.amini@centralesupelec.fr}}%
}

\begin{document}

\maketitle
\thispagestyle{empty}
\pagestyle{empty}

%%%%%%%%%%%%%%%%%%%%%%%%%%%%%%%%%%%%%%%%%%%%%%%%%%%%%%%%%%%%%%%%%%%%%%%%%%%%%%%%
\begin{abstract}

In this paper, we study the ergodic theorem for infinite-dimensional quantum Markov semigroups, originally introduced by Frigerio and Verri in 1982, and its latest version developed by Carbone and Girotti in 2021. We provide a sufficient condition that ensures exponential convergence towards the positive recurrent subspace, a well-known result for irreducible quantum Markov semigroups in finite-dimensional Hilbert spaces. Several illustrative examples are presented to demonstrate the application of the ergodic theorem. Moreover, we show that the positive recurrent subspace plays a crucial role in the study of global asymptotic stability.

\end{abstract}

%%%%%%%%%%%%%%%%%%%%%%%%%%%%%%%%%%%%%%%%%%%%%%%%%%%%%%%%%%%%%%%%%%%%%%%%%%%%%%%%
\section{Introduction}

In Section \ref{Section_Prel_Notions}, we first present the central notions for studying mean-ergodic QMS. In Section \ref{Section_Ergodic_Thm}, we consider the ergodic theorem by Frigerio and Verri \cite{frigerio1982long} for infinite-dimensional Hilbert spaces. We then present its reformulation given in \cite{carbone2021absorption}, where the authors show that a QMS is ergodic when the positive recurrent subspace is absorbing. Assuming the orthogonal of the positive recurrent subspace is finite-dimensional, we prove that the convergence happens exponentially fast (Theorem \ref{Thm_Exp_Conv}). This exponential behavior arises from the semigroup structure.

In Section \ref{Section_Applications}, we present three physical examples and show that the associated QMS is ergodic. In the first example, the two-photon absorption and emission model \cite{fagnola2005two}, we use the fixed-point structure of the semigroup and of its adjoint to demonstrate the ergodicity of the QMS (Theorem \ref{Thm_2_Photon_Abs}). For a generic QMS \cite{carbone2014asymptotic}, we recall that in certain cases, such as when the diagonal of the QMS can be represented by a birth and death process, classical probabilistic arguments can be applied to establish ergodicity. The third example involves Schrödinger cats stabilization via reservoir engineering \cite{azouit2016well, robin2024convergence}, where the authors in \cite{azouit2016well} construct a Lyapunov function. We show that it implies that the associated positive recurrent subspace is absorbing (Proposition \ref{Prop_Q_Harmonic_Oscill}).

In Section \ref{Section_Decomp}, we present the decomposition of the positive recurrent subspace into minimal enclosures, adapted to infinite-dimensional separable Hilbert spaces, building on the framework developed in \cite{baumgartner2012structures, carbone2016irreducible}. Drawing a parallel with the finite-dimensional decompositions proposed in \cite{cirillo2015decompositions}, we show that the positive recurrent subspace remains the minimal \textit{globally asymptotically stable} (\textit{GAS}) subspace even in the infinite-dimensional setting (Theorem \ref{Thm_GAS}). Finally, in Section \ref{Section_Finite_Dim}, we recall that for irreducible QMS on finite-dimensional Hilbert spaces, the rate of convergence can be made explicit \cite{wolf2012quantum}, and we provide a detailed proof of this result.

%%%%%%%%%%%%%%%%%%%%%%%%%%%%%%%%%%%%%%%%%%%%%%%%%%%%%%%%%%%%%%%%%%%%%%%%%%%%%%%%

\section{Preliminary notions} \label{Section_Prel_Notions}

\subsection{Quantum Markov semigroup}

In this paper, we will consider an infinite-dimensional separable Hilbert space $\mathcal{H}$ and we denote by $\mathcal{B(H)}$ the von Neumann algebra of bounded operators acting on $\mathcal{H}$.

Following \cite{fagnola2004quantum}, a quantum Markov semigroup can be defined as follows:
\begin{definition}
    A QMS on $\mathcal{B(H)}$ is a family $\Phi = \left( \Phi_t \right)_{t\ge0}$ of bounded operators on $\mathcal{B(H)}$ with the following properties:
    \begin{enumerate}
        \item $\Phi_0$ is the identity operator on $\mathcal{B(H)}$;
        \item $\Phi_{t+s} = \Phi_t \circ \Phi_s$ for all $t,s \ge 0$;
        \item $\Phi_t$ is completely positive for all $t\ge0$;
        \item $\Phi_t$ is a normal operator for all $t\ge0$, see \cite{fagnola1999quantum} Proposition 1.15 for details;
        \item for all $A \in \mathcal{B(H)}$, the map $t \to \Phi_t (A)$ is weak* continuous;
        \item $\Phi_t \left( \mathds{1} \right) = \mathds{1}$ for all $t\ge0$, where $\mathds{1}$ denotes the identity operator on $\mathcal{B(H)}$. 
    \end{enumerate}
\end{definition}

If the last item in the definition is omitted, $\Phi$ is called a \textit{quantum dynamical semigroup} (\textit{QDS}). In the special case of uniformly continuous QMS, \textit{i.e}, of QMS $\Phi$ that satisfies $\lim_{t \to 0} || \Phi_t - \Phi_0|| = 0$, Lindblad proved the well-known theorem:

\begin{theorem} [Lindblad]
    A bounded operator $\mathcal{L}$ on $\mathcal{B(H)}$ is the infinitesimal generator of a uniformly continuous QMS iff there exists a complex separable Hilbert space $\mathcal{K}$, a bounded operator $L : \mathcal{H} \to \mathcal{H} \otimes \mathcal{K}$ and an operator $G$ in $\mathcal{H}$ such that:
    \begin{equation*}
        \mathcal{L} (A) = L^* \left( A \otimes \mathds{1}_\mathcal{K} \right) L + G^* A + A G
    \end{equation*}
    for all $A \in \mathcal{B(H)}$.
\end{theorem}

Concretely, this means that for all $t\ge0$, we have:
\begin{equation*}
    \Phi_t = e^{t \mathcal{L}} = \sum_{n\ge0} \frac{t^n}{n!} \mathcal{L}^n
\end{equation*}
with the series being uniformly convergent.

However, most of Lindbladians encountered in infinite-dimensional physical models are unbounded, so Lindblad's theorem cannot be used. In this case, Chebotarev developed a theory to solve the equation that characterises a generalised Lindbladian:
\begin{equation} \label{Eq_Generalized_QMS}
    \langle v, \Phi_t (A) u \rangle = \langle v, Au \rangle + \int_0^t \langle v, \mbox{\st{$\mathcal{L}$}} \left( \Phi_s (A) \right) u \rangle ds
\end{equation}
for $A \in \mathcal{B(H)}$ and $u,v \in D(G)$ where $D(G)$ is the domain of $G$ and with
\begin{equation} \label{Eq_Generalized_Lindbladian}
    \langle v, \mbox{\st{$\mathcal{L}$}} (A) u \rangle = \langle v, AGu \rangle + \langle Gv, Au \rangle + \sum_{l = 1}^\infty \langle L_l v, A L_l u \rangle
\end{equation}
assuming that $D(G) \subseteq D(L_l)$ for all $l\ge1$. The following theorem ensures the existence of a \textit{QDS} solution of (\ref{Eq_Generalized_QMS}), which is called the \textit{minimal QDS}:

\begin{theorem} [\cite{fagnola1999quantum}, Theorem 3.22] \label{Thm_Existence_Minimal_QDS}
    Assume that $G$ is the infinitesimal generator of a strongly continuous contraction semigroup and that for every $u \in \mathcal{D(G)}$,
    \begin{equation*}
        \langle u, Gu \rangle + \langle Gu, u \rangle + \sum_{l=1}^\infty \langle L_l u, L_l u \rangle \le 0
    \end{equation*}
    Then there exists a QDS $\Phi_t^{(\mathrm{min})}$ solution of (\ref{Eq_Generalized_QMS}) with the following properties:
    \begin{enumerate}
        \item $\Phi_t^{(\mathrm{min})} \left( \mathds{1} \right) \le \mathds{1}$ for all $t\ge0$;
        \item for every weak* continuous family $\left( \Phi_t \right)_{t\ge0}$ of positive maps on $\mathcal{B(H)}$ solution of (\ref{Eq_Generalized_QMS}) and every positive operator $A \in \mathcal{B(H)}$, we have: $\Phi_t^{(\mathrm{min})} (A) \le \Phi_t (A)$
    \end{enumerate}
    Moreover, if $\Phi_t^{(\mathrm{min})}$ is Markov, meaning identity preserving, then it is the unique solution of (\ref{Eq_Generalized_QMS}).
\end{theorem}

Note that the assumption on $G$ can be verified using the Hille-Yosida theorem or the Lumer-Philips theorem, see for example \cite{robin2024convergence}. Moreover, there exist some criteria to assess if the minimal QDS is Markov or not, see for example Theorem 3.40 in \cite{fagnola1999quantum}.

\subsection{Transience and recurrence of QMS}

QMS being the analogue of Markov semigroups in classical probabilities, it is natural to define a decomposition of the Hilbert space into recurrent and transient parts for QMS, as is the case for Markov semigroups. This is very useful to qualitatively analyse the large-time behavior of the related open quantum system. The \textit{positive recurrent subspace} has been identified for a long time (\cite{frigerio1982long}) and is defined as follows:
\begin{equation*}
    \mathcal{R}_+ := \sup \left\{ \supp (\rho), \rho \ \mathrm{ normal \ invariant \ state \ of }\  \Phi \right\}
\end{equation*}
where $\supp$ denotes the support. \textit{Normal states} are defined for example in \cite{fagnola1999quantum}, see Definition 1.13.

Then, in \cite{fagnola2003transience}, Fagnola and Rebolledo introduced the notion of \textit{form-potential} of a positive bounded operator $A \in \mathcal{B(H)}$:
\begin{equation*}
    \mathfrak{U}(A)[v] = \int_0^\infty \langle v, \Phi_t (A) v \rangle dt
\end{equation*}
for any $v \in \mathcal{H}$ such that the integral is finite. For every $A$ such that $\mathfrak{U}(A)$ is bounded, it is represented by a self-adjoint operator $\mathcal{U}(A)$, see \cite{fagnola2004quantum}. Then, the \textit{transient subspace} can be defined as follows:
\begin{equation*}
    \mathcal{T} := \sup \left\{ \supp (\mathcal{U} (A)), A \ \mathrm{s. \ t. \ \mathcal{U}(A) \ bounded} \right\}
\end{equation*}

Finally, the remaining subspace is called the \textit{null recurrent subspace}:
\begin{equation*}
    \mathcal{R}_0 := \mathcal{R}_+^\perp \cap \mathcal{T}^\perp
\end{equation*}

And $\mathcal{R} := \mathcal{R}_+ \cup \mathcal{R}_0$ is called the \textit{recurrent subspace}. These subspaces play an important role in the study of the ergodic nature of a QMS.

With this introduction to QMS and its transient and recurrent subspaces, we are ready to present the ergodic theorem.

%%%%%%%%%%%%%%%%%%%%%%%%%%%%%%%%%%%%%%%%%%%%%%%%%%%%%%%%%%%%%%%%%%%%%%%%%%%%%%%%%%%%%%%%%%%%%%%%%%%

\section{Ergodic theorem} \label{Section_Ergodic_Thm}

Since a QMS is a semigroup, the question of its ergodic nature refers to the question whether it is mean ergodic or not. For a complete study of mean ergodic semigroups, see \cite{engel2000one}. To answer this question, Frigerio and Verri proved in 1982 the following ergodic theorem for QMS. Let $P_{\mathcal{R}_+}$ be the projection onto $\mathcal{R}_+$, and $\mathcal{F}(\Phi)$ and $\mathcal{F}(\Phi^*)$ the fixed points of the QMS $\Phi$ and of its adjoint semigroup $\Phi^*$.
\begin{theorem}[\cite{frigerio1982long}, Theorem 2.1] \label{Thm_Ergodic}
    For a QMS $\Phi$, the following are equivalent:
    \begin{enumerate}
        \item $\displaystyle w^* - \lim_{t\to\infty} \frac{1}{t} \int_0^t \Phi_s \left(P_{\mathcal{R}_+} \right) ds = \mathds{1}$;
        \item There exists a normal $\Phi$-invariant norm one projection $\mathcal{E}$ of $\mathcal{B(H)}$ onto $\mathcal{F}(\Phi)$;
        \item $\displaystyle \mathcal{E}^* (X) :=  w-\lim_{t\to\infty} \frac{1}{t} \int_0^t \Phi^*_s (X) ds$ \ exists for every $X \in L^1 (\mathcal{H})$;
        \item $\mathcal{F}(\Phi^*)$ separates $\mathcal{F}(\Phi)$.
    \end{enumerate}
    If the above conditions are satisfied, then for all $A \in \mathcal{B(H)}$,
    \begin{equation*}
        \mathcal{E}(A) := w^* - \lim_{t\to\infty} \frac{1}{t} \int_0^t \Phi_s (A) ds
    \end{equation*}
\end{theorem}

\begin{remark}
    The last item of the theorem simply means:
    \begin{equation*}
        \forall A,B \in \mathcal{F}(\Phi), \exists X \in \mathcal{F}(\Phi^*), \langle X, A \rangle \ne \langle X, B \rangle
    \end{equation*}
\end{remark}

Recently, in \cite{carbone2021absorption}, Carbone and Girotti proposed a new formulation of the first item, using the notion of \textit{absorption operators} that they developed. It is known that $\mathcal{R}_+$ is an \textit{enclosure}, \textit{i.e}, an invariant subspace, see for instance \cite{umanita2006classification}, which implies that $P_{\mathcal{R}_+}$ is subharmonic for $\Phi$:
\begin{equation*}
    P_{\mathcal{R}_+} \le \Phi_t \left( P_{\mathcal{R}_+} \right) \le \mathds{1} , \ \forall t \ge 0
\end{equation*}
Denote $A \left( \mathcal{R}_+ \right)$ the weak$^*$ limit of $\left( \Phi_t \left( P_{\mathcal{R}_+} \right) \right)_{t\ge0}$. It is the absorption operator associated with $\mathcal{R}_+$. And it is immediate that the first item of the theorem is equivalent to:
\begin{equation} \label{Eq_Equivalence_with_Abs_Operator}
    A \left( \mathcal{R}_+ \right) = \mathds{1}
\end{equation}
In other words, this means that $\mathcal{R}_+$ is \textit{attractive}. In particular, it implies that $\mathcal{R}_0$ is null and that $\mathcal{T}$ can be forgotten in large-time studies. This is formalized in the next proposition.
\begin{proposition} \label{Prop_R_attractive}
    Assume that $A \left( \mathcal{R}_+ \right) = \mathds{1}$. Then for any normal state $\rho$, as $t$ goes to infinity, $\Phi_t^* (\rho)$ will only be supported by $\mathcal{R}_+$.
\end{proposition}
\begin{proof}
    First, since $\mathds{1} - P_{\mathcal{R}_+}$ is bounded and since $\rho$ and $\Phi_t^* (\rho)$ are trace-class for any $t\ge0$, we have the following:
    \begin{equation*} \begin{split}
        \tr \left( \Phi_t \left( \mathds{1} - P_{\mathcal{R}_+} \right) \rho \right) &= \tr \left( \left( \mathds{1} - P_{\mathcal{R}_+} \right) \Phi_t^* (\rho) \right)\\ &= \tr \left( \Phi_t^* (\rho) \left( \mathds{1} - P_{\mathcal{R}_+} \right) \right)\\  &= \tr \left( \left( \mathds{1} - P_{\mathcal{R}_+} \right) \Phi_t^* (\rho) \left( \mathds{1} - P_{\mathcal{R}_+} \right) \right)
    \end{split} \end{equation*}
    Therefore, by linearity,
    \begin{equation*}
        \lim_{t\to\infty} \tr \left( \left( \mathds{1} - P_{\mathcal{R}_+} \right) \Phi_t^* (\rho) \left( \mathds{1} - P_{\mathcal{R}_+} \right) \right) = 0
    \end{equation*}
    Since $\left( \mathds{1} - P_{\mathcal{R}_+} \right) \Phi_t^* (\rho) \left( \mathds{1} - P_{\mathcal{R}_+} \right)$ is positive, this implies that it goes to zero as t goes to infinity. Lemma 2.1 in \cite{fagnola2004quantum} concludes the proof.
\end{proof}

\vspace{2mm}

Moreover, due to the structure of semigroup, assuming that $\mathcal{R}_+$ is attractive and that $\mathcal{R}_+^\perp$ is finite-dimensional, we can demonstrate that the convergence towards $\mathcal{R}_+$ is exponentially fast:
\begin{theorem} \label{Thm_Exp_Conv}
    Assume that $A \left( \mathcal{R}_+ \right) = \mathds{1}$ and that $\mathcal{R}_+^\perp$ is finite-dimensional. Then, $\left( \Phi_t \left( \mathds{1} - P_{\mathcal{R}_+} \right) \right)_{t\ge0}$ goes exponentially fast to zero. As a result, given a normal state $\rho$, $\left( \Phi_t^* (\rho) \right)_{t\ge0}$ becomes exponentially fast supported by $\mathcal{R}_+$.
\end{theorem}
\begin{proof}
    First, by linearity of $\Phi_t$ and because it is Markov, $A \left( \mathcal{R}_+ \right) = \mathds{1}$ is equivalent to
    \begin{equation*}
        \lim_{t\to\infty} \Phi_t \left( \mathds{1} - P_{\mathcal{R}_+} \right) = 0
    \end{equation*}
    Note also that the convergence towards zero is monotone decreasing. Moreover, $\Phi_t \left( \mathds{1} - P_{\mathcal{R}_+} \right)$ is only supported by $\mathcal{R}_+^\perp$, for any $t\ge0$. This is an application of Lemma 2.1 in \cite{fagnola2004quantum} together with the inequality:
    \begin{equation*}
        0 \le P_{\mathcal{R}_+} \Phi_t \left( \mathds{1} - P_{\mathcal{R}_+} \right) P_{\mathcal{R}_+} \le P_{\mathcal{R}_+} \left( \mathds{1} - P_{\mathcal{R}_+} \right) P_{\mathcal{R}_+} = 0
    \end{equation*}
    
    $\mathcal{R}_+^\perp$ being finite-dimensional, this gives the existence of a constant $0 \le \kappa < 1$ and a $t_0 > 0$ such that:
    \begin{equation} \label{Eq_Proof_Exp_Conv}
        0 \le \Phi_{t_0} \left( \mathds{1} - P_{\mathcal{R}_+} \right) \le \kappa \left( \mathds{1} - P_{\mathcal{R}_+} \right)
    \end{equation}

    Now, using the fact that $\Phi$ is a semigroup of positive operators, Equation (\ref{Eq_Proof_Exp_Conv}) can be expressed in the following terms: for any $t\ge0$,
    \begin{equation*}
        \Phi_t \left( \mathds{1} - P_{\mathcal{R}_+} \right) \le \kappa^{\left\lfloor t/t_0 \right\rfloor} \left( \mathds{1} - P_{\mathcal{R}_+} \right)
    \end{equation*}
    thus the exponential convergence.
\end{proof}

%%%%%%%%%%%%%%%%%%%%%%%%%%%%%%%%%%%%%%%%%%%%%%%%%%%%%%%%%%%%%%%%%%%%%%%%%%%%%%%%

\section{Applications of the ergodic theorem} \label{Section_Applications}

The ergodic theorem \ref{Thm_Ergodic} provides some very useful information about the qualitative behavior of the QMS in large time. However, in general it is not possible to explicitly verify the condition $A(\mathcal{R}_+) = \mathds{1}$ in infinite dimension. In the following, we present some physical examples and prove that the associated QMS is ergodic. 

\subsection{Two-photon absorption and emission model}

The two-photon absorption and emission process is a simple model that already shows a rich behavior. Following \cite{fagnola2005two}, we define it as follows. Consider the Hilbert space $l^2 (\mathbb{N})$ with its canonical basis $\left( e_n \right)_{n\ge0}$ and let $a, a^\dag$ and $N$ be the annihilation, creation and number operators. As in Equation (\ref{Eq_Generalized_Lindbladian}), consider the generalised Lindbladian
\begin{equation} \label{Eq_Lindblad_2_Photon}
    \langle v, \mbox{\st{$\mathcal{L}$}} (A) u \rangle = \langle v, AGu \rangle + \langle Gv, Au \rangle + \sum_{l = 1}^2 \langle L_l v, A L_l u \rangle
\end{equation}
for any $u,v \in D(G) = D(N^2)$ where
\begin{equation*}
    G = - \frac{\lambda^2}{2} a^2 a^{\dag 2} - \frac{\mu^2}{2} a^{\dag 2} a^2 - \imath \omega a^{\dag 2} a^2
\end{equation*}
\begin{equation*}
    L_1 = \mu a^2 \ \ \mathrm{and} \ \ L_2 = \lambda a^{\dag 2}
\end{equation*}
with $\lambda \ge 0, \mu > 0$ and $\omega \in \mathbb{R}$. It has been shown in \cite{fagnola2005two} that if $\nu = \lambda /\mu \le 1$, then Theorem \ref{Thm_Existence_Minimal_QDS} is satisfied, so there exists a unique QMS generated by the Lindbladian $\mathcal{L}$ associated to (\ref{Eq_Lindblad_2_Photon}).

\vspace{2mm}

\begin{theorem} \label{Thm_2_Photon_Abs}
    If $\lambda/\mu < 1$, then the positive recurrent subspace of the QMS associated to the two-photon absorption and emission model is attractive.
\end{theorem}
\begin{proof}
Let us first assume that $\lambda > 0$. Then, by Theorem 6.1 in \cite{fagnola2005two}, the only invariant states of $\Phi^*$ are
\begin{equation*}
    \rho_e = \left( 1 - \nu^2 \right) \sum_{k\ge0} \nu^{2k} | e_{2k} \rangle \langle e_{2k} |
\end{equation*}
\begin{equation*}
    \rho_o = \left( 1 - \nu^2 \right) \sum_{k\ge0} \nu^{2k} | e_{2k+1} \rangle \langle e_{2k+1} |
\end{equation*}
and any convex combination of $\rho_e$ and $\rho_o$. Therefore, there exists at least one faithful invariant state, which implies that $\mathcal{R}_+ = l^2 (\mathbb{N})$. Then, as $P_{\mathcal{R}_+} = \mathds{1}$, we obviously have $A (\mathcal{R}_+) = \mathds{1}$. So, if $\lambda > 0$, the QMS is ergodic.

\vspace{2mm}

Suppose now that $\lambda=0$, \textit{i.e} suppose that there is no emission. Then, by Proposition 3.2 in \cite{fagnola2005two}, all the invariant states have the form 
\begin{equation*}
    \alpha \rho_e + (1-\alpha) \rho_o + z |e_0 \rangle \langle e_1 | + \Bar{z} |e_1 \rangle \langle e_o|
\end{equation*}
with $\alpha \in [0,1]$ and $|z|^2 \le \alpha (1 - \alpha)$. Therefore, in this case, $\mathcal{R}_+ = \Span \left\{ |e_0\rangle, |e_1\rangle \right\}$. To study the ergodic properties of the QMS, we look at the set of fixed points of $\Phi$. $A \in \mathcal{B(H)}$ is in $\mathcal{F}(\Phi)$ if and only if $\mathcal{L} (A) = 0$, \textit{i.e}, if and only if: 
\begin{itemize}
    \item $A_{j,k}=0$ if $j\ge2$ and $k\le1$ or if $j\le1$ and $k\ge2$; 
    \item for any $j,k \ge 2$,
\end{itemize}
\begin{equation*} \begin{split}
    ( \imath \omega (j(j-1) - k(k-1)) - \frac{\mu^2}{2} (k(k-1) + j(j-1)) ) A_{j,k} \\
    + \ \mu^2 \sqrt{j(j-1)k(k-1)} A_{j-2,k-2} = 0 \phantom{hhhhhh}
\end{split} \end{equation*}
Observe that due to the induction relations, $A$ is totally determined by $A_{0,0}, A_{1,0}, A_{0,1}$ and $A_{1,1}$.
Then, if $\rho$ is an invariant state,
\begin{equation*}
    \langle \rho, A \rangle = \tr (\rho A) = (1-\alpha) A_{0,0} + z A_{1,0} + \Bar{z} A_{0,1} + \alpha A_{1,1}
\end{equation*}
Thus, if $A$ is non-zero, there always exists an invariant state $\rho$ such that $\langle \rho, A \rangle \ne 0$. As $A$ can be viewed as the difference of two fixed points of $\Phi$, this implies that the set $\mathcal{F}(\Phi^*)$ separates $\mathcal{F}(\Phi)$. The QMS is also ergodic if $\lambda=0$.
\end{proof}

\begin{remark}
    Note that the fact that $\mathcal{F}(\Phi^*)$ separates $\mathcal{F}(\Phi)$ can be shown even faster. In fact, since $P_{\mathcal{R}_+} = | e_0 \rangle \langle e_0 | + | e_1 \rangle \langle e_1 |$, it is obvious from the induction relations that $P_{\mathcal{R}_+} A P_{\mathcal{R}_+} = 0$ implies $A = 0$. Therefore, $\mathcal{F}(\Phi^*)$ separates $\mathcal{F}(\Phi)$, see the proof of Theorem 2.1 in \cite{frigerio1982long} for a full explanation.
\end{remark}

\subsection{Generic QMS}

In the previous subsection, we used the characterization by the sets of fixed points to verify that the ergodic theorem applies. There is also a probabilistic approach to study the ergodic nature of the QMS. The idea is to find an abelian subalgebra of $\mathcal{B(H)}$ that is left invariant by the semigroup. This way, the restriction of the semigroup to this subalgebra determines a classical Markov semigroup and the classical theory can thus be applied. For example, we refer to the following proposition.
\begin{proposition} [\cite{fagnola2003transience}, Proposition 7]
    If the restriction of the semigroup to this subalgebra is transient, meaning that $\mathcal{T} = \mathcal{H}$, then so is $\Phi$.
\end{proposition}

This can be put in parallel with a special case of \textit{open quantum random walks} (\textit{OQRW}) called \textit{minimal OQRW} which can be thought of as diagonal OQRW, see \cite{mousset2025decomposition} for example. In \cite{carbone2014asymptotic}, the authors put this concept into practice for \textit{generic QMS}. A generic QMS is defined as follows. For $A \in D(G) \subseteq \mathcal{B(H)}$, consider the generalised Lindbladian on the Hilbert space $l^2 (\mathbb{N})$:
\begin{equation*}
    \mathcal{L}(A) = AG + G^* A + \sum_{j\ne m} L_{mj}^* A L_{mj}
\end{equation*}
with
\begin{equation*}
    G = - \sum_{m\ge0} \left( - \frac{\gamma_{mm}}{2} + \imath \kappa_m \right) |e_m\rangle \langle e_m |
\end{equation*}
\begin{equation*}
    L_{m,j} = \sqrt{\gamma_{mj}} | e_j \rangle \langle e_m |
\end{equation*}
where $\gamma_{mj} \ge 0$ for $m\ne j$, $\gamma_{mm} = -\sum_{j\ne m} \gamma_{mj} > - \infty$ and $\kappa_m \in \mathbb{R}$.

In \cite{carbone2007generic}, the authors proved that it can generate a minimal QMS under suitable conditions. The generic QMS has the very good property to map the diagonal part onto itself and the off-diagonal part onto itself as well. Moreover, when restricted to the diagonal, it generates a time continuous Markov chain $\left( X_t \right)_{t\ge0}$ with generator $\Gamma$ defined as:
\begin{equation*}
    \mathcal{L} \left( \sum_n f(n) | e_n\rangle \langle e_n| \right) = \sum_n ( \Gamma f ) (n) |e_n\rangle \langle e_n |
\end{equation*}
for any $f \in l^\infty (\mathbb{N})$ with $( \Gamma f ) (n) = \sum_j \gamma_{nj} f(j)$. 
\begin{example} \label{Ex_Birth_Death_Process}
    A classical example of time continuous Markov chain are birth and death processes. Consider the infinitesimal generator $\Gamma = \left( \gamma_{i,j} \right)_{i,j\ge0}$ and its associated transition matrix $P = \left( p_{i,j} \right)_{i,j\ge0}$ with $p_{i,j} = - \gamma_{i,j}/\gamma_{i,i}$ if $i\ne j$ and $\gamma_{i,i} < 0$, $p_{i,i} = 0$ if $\gamma_{i,i} < 0$ and $p_{i,j} = \delta_{i,j}$ if $\gamma_{i,i}=0$. It is a birth and death process if $p_{i,j} = 0$ as soon as $|i-j|\ge2$.
\end{example}

The following Lemma and Theorem shows how to use the classical probability theory to demonstrate that a generic QMS is ergodic:

\begin{lemma}[\cite{carbone2014asymptotic}, Lemma 2 (a)] \label{Lemma_Generic_QMS}
    Consider an off-diagonal operator $Z = \sum_{k\ne m} Z_{km} |e_k\rangle \langle e_m|$. If there exists $j \in \mathbb{N}$ such that $\gamma_{kk} \ne 0$ for all $k \ne j$, then $\lim_{t\to\infty} \Phi_t (Z) = 0$ and it happens exponentially fast.
\end{lemma}
As a result, we have:
\begin{theorem}[\cite{carbone2014asymptotic}, Theorems 15 \& 16]
    If $\left( X_t \right)_{t\ge0}$ is positive recurrent and if the condition in Lemma \ref{Lemma_Generic_QMS} holds, then the QMS is ergodic.
\end{theorem}

\begin{example} [following Example \ref{Ex_Birth_Death_Process}]
    Assume that the birth and death process is irreducible, which is equivalent to $p_{i,i+1} p_{i+1,i} >0$ for any $i \in \mathbb{N}$. Note that this also implies the assumption in Lemma \ref{Lemma_Generic_QMS}. Suppose that it has no finite explosion time, which means, using Reuter's criteria, that:
    \begin{equation*}
        \sum_{i\ge0} \frac{1}{\gamma_{i,i+1}} + \frac{\gamma_{i,i-1}}{\gamma_{i,i+1} \gamma_{i-1,i}} + \cdots + \frac{\gamma_{i,i-1} \cdots \gamma_{1,0}}{\gamma_{i,i+1}  \cdots \gamma_{0,1}} = \infty
    \end{equation*}
    Then, it is a classical exercise in probability to show that $\left( X_t \right)_{t\ge0}$ is positive recurrent if and only if
    \begin{equation*}
        \sum_{n\ge1} \prod_{i=1}^n \frac{p_{i,i-1}}{p_{i,i+1}} = \infty \ \ \mathrm{and} \ \ \sum_{n\ge1} \prod_{i=1}^n \frac{p_{i-1,i}}{p_{i,i-1}} < \infty
    \end{equation*}
    This way, we can directly infer the ergodic nature of the related generic QMS.
\end{example}

\subsection{Quantum harmonic oscillator}

A third and more direct way to prove that a QMS is ergodic is to find a Lyapunov function that will directly give the convergence towards $\mathcal{R}_+$. To illustrate this, let's consider a quantum harmonic oscillator with k-photon exchange as in \cite{azouit2016well}. This model is very useful because it can be used to stabilize entangled states through reservoir engineering. Again, the Hilbert space is $l^2(\mathbb{N})$. The generalised Lindbladian in the Schrödinger picture is defined as follows:
\begin{equation*}
    \mathcal{L}^*(\rho) = L \rho L^\dag - \frac{1}{2} L^\dag L \rho - \frac{1}{2} \rho L^\dag L
\end{equation*}
where $L = a^k - \alpha^k \mathds{1}$ with $\alpha\in\mathbb{R}$. In \cite{azouit2016well}, the authors proved that the Hille-Yosida theorem holds for this case which gives the existence of a semigroup $\Phi^*$ generated by $\mathcal{L}^*$. Moreover, using the Lyapunov function 
\begin{equation*}
    V(\rho) = \tr \left( L \rho L^\dag \right)
\end{equation*}
they proved that any normal state will be absorbed in large time by $\ker (L)$, a finite-dimensional subspace spanned by the \textit{coherent states}, see subsection 4.3 in \cite{robin2024convergence}. This way, we can deduce the following proposition:
\begin{proposition} \label{Prop_Q_Harmonic_Oscill}
    For the model of the quantum harmonic oscillator with k-photon exchange, we have:
    \begin{equation*}
        \mathcal{R}_+ = \ker (L) \ \ \mathrm{and} \ \ A (\mathcal{R}_+ ) = \mathds{1}
    \end{equation*}
\end{proposition}
\begin{proof}
    Each state in $\ker (L)$ is an invariant state, so by definition $\ker (L) \subseteq \mathcal{R}_+$. Moreover, as for any normal invariant state, say $\rho$, $\lim_{t\to\infty} V(\Phi_t^* (\rho)) = 0$, this implies:
    \begin{equation*}
        V(\rho) = \lim_{t\to\infty} V(\Phi_t^* (\rho)) = 0
    \end{equation*}
    and necessarily $\mathcal{R}_+ \subseteq \ker (L)$. Thus the equality.

    On the other hand, for any finite-dimensional state $\rho$,
    \begin{equation*} \begin{split}
        0 = \lim_{t\to\infty} V(\Phi_t^* (\rho)) &= \lim_{t\to\infty} \tr \left( L^\dag L \Phi_t^*(\rho) \right)
    \end{split} \end{equation*}
    and this can be extended to any normal state by density. Moreover, the spectrum of $\sqrt{L^\dag L}$ has no accumulation point at zero, see \cite{azouit2016well}, which implies that there exists a constant $c\ge0$ such that $\mathds{1} - P_{\mathcal{R}_+} \le c L^\dag L$. Therefore,
    \begin{equation*} \begin{split}
        \tr \left( \Phi_t \left( \mathds{1} - P_{\mathcal{R}_+} \right) \rho \right) &= \tr \left( \left( \mathds{1} - P_{\mathcal{R}_+} \right) \Phi_t^* (\rho) \right) \\ &\le c \tr \left( L^\dag L \Phi_t^* (\rho) \right)
    \end{split} \end{equation*}
    So, $w^* - \lim_{t\to\infty} \tr \left( \Phi_t \left( \mathds{1} - P_{\mathcal{R}_+} \right) \rho \right) = 0$. Finally, because of uniform boundedness of $\left( \Phi_t \left( \mathds{1} - P_{\mathcal{R}_+} \right) \right)_{t\ge0}$ (by one) and weak$^*$-compactness of the balls of $\mathcal{B(H)}$,    
    \begin{equation*}
        w^*-\lim_{t\to\infty} \Phi_t \left( \mathds{1} - P_{\mathcal{R}_+} \right) = 0. 
    \end{equation*}
\end{proof}

\vspace{2mm}

However, it can be very difficult to find a Lyapunov function in general, as discussed in \cite{robin2024convergence}, where the authors study a model inspired by the quantum harmonic oscillator but for which is it much harder to find a suitable Lyapunov function. Nevertheless, the ergodic theorem can provide useful insights to find a good one. For example, sometimes, the function $\rho \to \tr \left( \left( \mathds{1} - P_{\mathcal{R}_+} \right) \rho \right)$ can already be a Lyapunov function that gives an explicit convergence speed.

\vspace{2mm}

In conclusion, we have studied three ways to show that $\mathcal{R}_+$ is attractive: using the fixed points of the two semigroups $\Phi$ and $\Phi^*$, using the restriction to the diagonal to apply classical probability and using a Lyapunov function to infer the result directly.

%%%%%%%%%%%%%%%%%%%%%%%%%%%%%%%%%%%%%%%%%%%%%%%%%%%%%%%%%%%%%%%%%%%%%%%%%%%%%%%%%%%%%%%

\section{Decomposition of the positive recurrent subspace} \label{Section_Decomp}

\subsection{Decomposition into minimal enclosures}

In the previous sections, we examined the convergence towards the positive recurrent subspace $\mathcal{R}_+$. A natural question then is: what happens inside this subspace? How does the semigroup behave? In particular, are there some subspaces inside $\mathcal{R}_+$ that are invariant with regard to the QMS, as is $\mathcal{R}_+$ \cite{carbone2021absorption}. These questions can be answered using the decomposition of $\mathcal{R}_+$ into minimal invariant enclosures developed by Baumgartner and Narnhofer for time continuous evolutions in finite-dimensional Hilbert spaces \cite{baumgartner2012structures}. Then, in \cite{carbone2016irreducible}, Carbone and Pautrat extended it to infinite-dimensional separable Hilbert spaces for discrete time evolution. We adapted their formalism, which is slightly different from the one in \cite{baumgartner2012structures}, in continuous time for finite-dimensional Hilbert spaces in \cite{mousset2025decomposition}. The generalisation to separable Hilbert spaces in infinite dimension is very similar and will be detailed in a future work.

\vspace{2mm}

Consider a QMS $\Phi$ on the separable Hilbert space $\mathcal{H}$. Before presenting the decomposition of its positive recurrent subspace $\mathcal{R}_+$, we need to clarify some notions.
\begin{definition}
    A subspace $\mathcal{V}$ of $\mathcal{H}$ is an \textit{enclosure} if for any normal state $\rho$ with support in $\mathcal{V}$, for all $t \in \mathbb{R}_+$, $\Phi_t (\rho)$ has also support in $\mathcal{V}$.
\end{definition}
Note that enclosures are also called \textit{invariant subspaces} in the literature. The decomposition of $\mathcal{R}_+$ is based on a certain type of enclosures called \textit{minimal enclosures} or \textit{minimal invariant subspaces}:
\begin{definition}
    A non-trivial enclosure is a \textit{minimal enclosure} if the only non-trivial enclosure it contains is itself.
\end{definition}
Enclosures and invariant states are closely related, since every enclosure inside of $\mathcal{R}_+$ is the support of an invariant state \cite{carbone2016irreducible}. In particular:
\begin{proposition}[\cite{carbone2016irreducible}, Proposition 5.7]
    A subspace of $\mathcal{R}_+$ is a minimal enclosure if and only if it is the support of an extremal invariant state, \textit{i.e}, of an extremal point of the convex set of normal invariant states of $\mathcal{R}_+$.
\end{proposition}

As in \cite{baumgartner2012structures} and \cite{carbone2016irreducible}, the decomposition is the following:

\begin{theorem} \label{Thm_Decomp_Min_Encl}
    There exists a decomposition of $\mathcal{R}_+$ of the form:
    \begin{equation} \label{Eq_Decomp_Hilbert_Space}
        \mathcal{R}_+ = \sum_{\alpha\in A} \mathcal{V}_a \oplus \sum_{\beta \in B} \sum_{\gamma \in C_\beta} \mathcal{V}_{\beta, \gamma}
    \end{equation}
    where $A$ or $B$ can be empty, $C_\beta$ contains at least two elements, and:
    \begin{itemize}
        \item All the $\mathcal{V}_\alpha$ and $\mathcal{V}_{\beta, \gamma}$ are mutually orthogonal minimal enclosures;
        \item Any minimal enclosure that is not orthogonal to $\sum_{\gamma \in C_\beta} \mathcal{V}_{\beta, \gamma}$ is contained in $\sum_{\gamma \in C_\beta} \mathcal{V}_{\beta, \gamma}$.
    \end{itemize}
\end{theorem}
\begin{remark}
    This decomposition is generally not unique. $A$ and the $\mathcal{V}_\alpha$ correspond to minimal enclosures that appear in each possible decomposition. On the other hand, $B$ and $\sum_{\gamma \in C_\beta} \mathcal{V}_{\beta, \gamma}$ will always appear as well, but the choice of the $\mathcal{V}_{\beta, \gamma}$ is not unique.
\end{remark}

We can apply this decomposition to the previous examples.
\begin{example}
    Consider again the two-photon absorption and emission model. If $\lambda > 0$, there are only two extremal invariant states: $\rho_e$ and $\rho_o$. Thus, the decomposition of $\mathcal{R}_+$ in this case simply is:
    \begin{equation*}
        \mathcal{R}_+ = \supp(\rho_e) \oplus \supp(\rho_o)
    \end{equation*}
    On the other hand, if $\lambda=0$, any rank one state with support in $\Span( |e_0 \rangle, |e_1\rangle)$ is invariant. Therefore, there is an infinite number of extremal invariant states and an infinite number of decompositions of $\mathcal{R}_+$. A possible one is:
    \begin{equation*}
        \mathcal{R}_+ = \Span( |e_0\rangle) \oplus \Span ( |e_1 \rangle)
    \end{equation*}
\end{example}

\begin{example} \label{Ex_Birth_Death_Decomp}
    Consider the birth and death process of Example \ref{Ex_Birth_Death_Process} and assume that it is irreducible positive recurrent with no finite explosion time. This way, it has exactly one invariant probability measure $\pi$ on $\mathbb{N}$ given by:
    \begin{equation*}
        \pi (0) = \frac{1}{S} \ \ \mathrm{and} \ \ \pi (n) = \frac{1}{S} \prod_{k=1}^n \frac{p_{k-1,k}}{p_{k,k-1}} \ \ \mathrm{for} \ \ n \ge 1
    \end{equation*}
    with $S = 1 + \sum_{n\ge1} \prod_{1\le k \le n}  \frac{p_{k-1,k}}{p_{k,k-1}}$. Then, by the third item of Theorem 16 in \cite{carbone2014asymptotic}, any associated generic QMS will have exactly one invariant state $\rho_{inv}$ which will be diagonal and such that $\langle e_n | \rho_{inv} | e_n \rangle = \pi (n)$. Therefore, the only invariant state is faithful. In this case, we say that the QMS is \textit{irreducible} and $\mathcal{R}_+ = l^2 (\mathbb{N})$ is the only possible decomposition of the positive recurrent subspace.
\end{example}

\begin{example}
    Finally, consider the quantum harmonic oscillator with k-photon exchange model. Since the generalised Lindbladian only depends on $L$ and since $\mathcal{R}_+ = \ker(L)$, we can directly see that any rank one state with support in $\ker(L)$ is an extremal invariant state. Therefore, there is again an infinite number of decompositions of $\mathcal{R}_+$. A possible one is:
    \begin{equation*}
        \mathcal{R}_+ = \bigoplus_{r=1}^k | \psi^r_L \rangle
    \end{equation*}
    where $\left\{ | \psi^r_L \rangle \right\}_{1\le r \le k}$ is an orthogonal basis of $\ker(L)$, see section 4.3 of \cite{robin2024convergence}.
\end{example}

\subsection{Link with the NFD and global asymptotic stability}

There exist other possible decompositions of the Hilbert space associated to a QMS. In finite dimension, Ticozzi and Cirillo proposed two other decompositions : the \textit{nested-face decomposition} (NFD) and the \textit{dissipation-induced decomposition} (\textit{DID}) \cite{cirillo2015decompositions}. We will only discuss about the first one, as it is close to the decomposition into minimal enclosures described earlier.

In discrete time, meaning that the adjoint semigroup is made of powers of a quantum channel $T$, the NFD is defined as follows: consider an invariant subspace $\mathcal{H}_{S_1} \subseteq \mathcal{H}$ and its orthogonal $\mathcal{H}_{R_1} = \mathcal{H}_{S_1}^\perp$. Denote $\mathfrak{H}_{R_1}^+$ the space of density matrices with support equal to $R_1$, $P_{R_1}$ the projection onto $\mathcal{H}_{R_1}$ and define $T_{R_1} := P_{R_1} T (\cdot) P_{R_1}$ acting on matrices with support in $R_1$. Then, define
\begin{equation*}
    \mathcal{D}_2 := \ker \left( \left( T_{R_1} - \sigma \left( T_{R_1} \right) \mathds{1}_{\mathcal{H}_{R_1}} \right)^{d_1^2} \right)
\end{equation*}
where $\sigma(T_{R_1})$ is the spectral radius of $T_{R_1}$ and $d_1$ the dimension of $\mathcal{H}_{R_1}$. Then, define $\mathcal{H}_{T_1}$ and $\mathcal{H}_{S_2}$ as
\begin{equation*}
    \mathcal{H}_{T_1} := \supp (D_2)
\end{equation*}
\begin{equation*}
    \mathcal{H}_{S_2} := \mathcal{H}_{S_1} \oplus \mathcal{H}_{T_1}
\end{equation*}
and repeat the process with $\mathcal{H}_{R_2} = \mathcal{H}_{S_2}^\perp$ and so on and so forth until the entire Hilbert space has been processed. Note that this construction can be adapted to continuous time by considering the spectrum of the Lindbladian generating the semigroup instead of the one of $T$. We will not discuss it here, as it is beyond the scope of this article. However, the following two propositions remain valid in the continuous-time setting. By construction, the NFD has the following property:
\begin{proposition}[\cite{cirillo2015decompositions}, Proposition 2]
    \begin{equation*}
        \sigma \left( T_{R_i} \right) > \sigma \left( T_{R_{i+1}} \right)
    \end{equation*}
\end{proposition}

In particular, the NFD is very relevant to discuss about \textit{global asymptotic stability} (\textit{GAS}) of an invariant subspace, which is defined as follows:
\begin{definition}
    An invariant subspace $\mathcal{H}_S$ is GAS if for any density matrix $\rho$,
    \begin{equation*}
        \lim_{t\to\infty} \left| \left| \Phi_t^* (\rho) - P_S \Phi_t^* (\rho) P_S \right| \right| = 0
    \end{equation*}
    with $P_S$ the projection onto $\mathcal{H}_S$.
\end{definition}
\begin{proposition}[\cite{cirillo2015decompositions}, Proposition 3]
    $\mathcal{H}_{S_1}$ is GAS if and only if $\sigma \left( T_{R_1} \right) < 1$. If it is not the case, $\mathcal{H}_{S_2}$ is the minimal GAS subspace containing $\mathcal{H}_{S_1}$.
\end{proposition}
This Proposition can be commented with regards to $\mathcal{R}_+$ and $A (\mathcal{R}_+)$. Indeed, we have the following result:
\begin{proposition} \label{Prop_R_GAS}
    If $\mathcal{H}$ is finite-dimensional, then $\mathcal{R}_+$ is the smallest subspace that can be GAS. Any other GAS subspace must contain $\mathcal{R}_+$.
\end{proposition}
\begin{proof}
    First, note that for finite-dimensional Hilbert spaces, $\mathcal{A}(\mathcal{R}_+) = \mathds{1}$ is always verified \cite{carbone2021absorption}.
    By Proposition \ref{Prop_R_attractive}, $\mathcal{R}_+$ is therefore GAS and any smaller GAS subspace must be contained in $\mathcal{R}_+$. By Theorem \ref{Thm_Decomp_Min_Encl}, $\mathcal{R}_+$ can be decomposed into an orthogonal sum of minimal enclosures. As every minimal enclosure is the support of an invariant state, none of them can be GAS, unless $\mathcal{R}_+$ is a minimal enclosure, as in Example \ref{Ex_Birth_Death_Decomp}.
\end{proof}

\vspace{2mm}

This way, we can see that $\mathcal{H}_{S_1}$ is GAS if and only if it contains $\mathcal{R}_+$. If it doesn't, then $\mathcal{H}_{T_1} = \mathcal{R}_+ \cap \mathcal{H}_{S_1}^\perp$, so that $\mathcal{H}_{S_2}$ can be the minimal GAS subspace containing $\mathcal{H}_{S_1}$. Note that Proposition \ref{Prop_R_GAS} can be directly generalised to infinite-dimensional separable Hilbert spaces:
\begin{theorem} \label{Thm_GAS}
    Suppose that $A (\mathcal{R}_+) = \mathds{1}$. Then $\mathcal{R}_+$ is the smallest GAS subspace and it is contained in any other GAS subspace.
\end{theorem}
\begin{proof}
    Just note that Proposition \ref{Prop_R_attractive} and Theorem \ref{Thm_Decomp_Min_Encl} still hold in this case.
\end{proof}

\vspace{2mm}

It is very difficult to find more precise statements about convergence towards invariant subspaces in all generality. Attraction domains of each minimal enclosure have to be taken into account, which can be done using absorption operators \cite{carbone2021absorption}. However, the transient subspace can distort the cards and we have to be careful about the fact that the decomposition in Theorem \ref{Thm_Decomp_Min_Encl} is not unique. Still, it can be done for specific examples, as in \cite{fagnola2005two} for the two-photon absorption and emission model. On the other hand, under additional assumptions, other results are available. For example, the following proposition holds:
\begin{proposition}
    Suppose that $\mathcal{T} = \mathcal{R}_0 = \{0\}$. Suppose also that there is only one decomposition of $\mathcal{R}_+$ in the form of Theorem \ref{Thm_Decomp_Min_Encl} and that each minimal enclosure is finite-dimensional. Then, for any normal state $\rho$ that is block-diagonal with respect to the decomposition,
    \begin{equation*}
        w^*-\lim_{t\to\infty} \Phi_t^* (\rho) = \sum_{\alpha \in A} \tr \left( P_\alpha \rho \right) \rho_\alpha
    \end{equation*}
    where $P_\alpha$ is the projection onto $\mathcal{V}_\alpha$ and $\rho_\alpha$ is the extremal invariant state associated to $\mathcal{V}_\alpha$.
\end{proposition}
\begin{proof}
    By definition, the restriction of $\Phi$ to a minimal enclosure $\mathcal{V}_\alpha$ is an irreducible QMS in finite dimension. Therefore, by Theorem \ref{Thm_Irred_QMS} in the next section,
    \begin{equation*}
        \lim_{t\to\infty} \Phi_t^* \left( P_\alpha \rho P_\alpha \right) = \tr \left( P_\alpha \rho \right) \rho_\alpha
    \end{equation*}
\end{proof}

Indeed, for irreducible QMS in finite-dimensional Hilbert spaces, we can have complete knowledge of the convergence of the QMS as well as the convergence speed.

%%%%%%%%%%%%%%%%%%%%%%%%%%%%%%%%%%%%%%%%%%%%%%%%%%%%%%%%%%%%%%%%%%%%%%%%%%%%%%%%%%%%%%%%%%%%

\section{Convergence speed for finite-dimensional Hilbert spaces} \label{Section_Finite_Dim}

The speed of convergence towards $\mathcal{R}_+$ or even to some normal invariant state is a very important question both in mathematics and for physical applications. A powerful tool to study it is the spectrum of the Lindbladian and, in particular, the question of the positivity of its spectral gap. This question can be very tricky and, in some cases, even undecidable \cite{cubitt2022undecidability}. However, there are models even in infinite-dimensional Hilbert spaces where it is possible to settle the question, see for example \cite{carbone2000exponential}. Restricting ourselves to finite-dimensional Hilbert spaces, we recall the strong result that for irreducible QMS, the spectral gap is strictly positive. We propose a detailed proof of this theorem, adapted from the one in \cite{wolf2012quantum}.

\begin{definition}
    A quantum channel is \textit{primitive} if it is irreducible and if 1 is its only eigenvalue of modulus one.
\end{definition}
\begin{theorem} [\cite{wolf2012quantum}, Proposition 7.5] \label{Thm_Irred_QMS}
    Let $\Phi^*$ be the adjoint semigroup of a QMS on a finite-dimensional Hilbert space. The followings are equivalent:
    \begin{enumerate}
        \item There is a $t_0 > 0$ such that $\Phi^*_{t_0}$ is irreducible;
        \item $\Phi_t^*$ is irreducible for all $t>0$;
        \item $\Phi_t^*$ is primitive for all $t>0$;
        \item There is a positive definite state $\rho_\infty$ such that for all state $\rho$, $\lim_{t\to\infty} \Phi_t^* (\rho) = \rho_\infty$;
        \item There is a positive definite state $\rho_\infty$ such that $\ker (\mathcal{L}^*) = \Span (\rho_\infty)$, $\mathcal{L}^*$ being the Lindbladian associated to $\Phi^*$.
    \end{enumerate}
\end{theorem}
\begin{proof}
    It is immediate to see that $3 \Rightarrow 2$ and $2 \Rightarrow 1$. Let's show that $1 \Rightarrow 3$. Since $\Phi_{t_0}^*$ is irreducible, by Theorem 6.6 in \cite{wolf2012quantum}, there exists an integer $m >0$ such that its matrix form in a certain basis is the following Jordan normal form:
    \begin{equation*}
        \widehat{\Phi}_{t_0}^* = \begin{pmatrix} D & 0 \\ 0 & T \end{pmatrix}
    \end{equation*}
    with $D := \diag\left(1, e^{ 2 \imath \pi / m} , \dots, e^{ 2 \imath (m-1) \pi / m } \right)$ and $T$ an upper triangular matrix with elements of modulus strictly less than one on its diagonal. In finite dimension, $\Phi^*$ is an exponential semigroup with generator $\mathcal{L}^*$, \textit{i.e}, $\Phi_{t_0}^* = e^{t_0 \mathcal{L}^*}$ \cite{engel2000one}, therefore we can reconstruct the matrix of $\mathcal{L}^*$. In the same basis,
    \begin{equation*}
        t_0 \widehat{\mathcal{L}}^* = \begin{pmatrix} D_{\mathcal{L}^*} & 0 \\ 0 & T_{\mathcal{L}^*} \end{pmatrix}
    \end{equation*}
    with $D_{\mathcal{L}^*} = \diag \left( 0, 2 \imath \pi/m, \dots, 2 \imath (m-1) \pi / m \right)$ defined modulus $2 \imath \pi$ and $T_{\mathcal{L}^*}$ an upper triangular matrix of which diagonal components have a negative real part. Although non unique, $\mathcal{L}^*$ has a one-dimensional kernel, and the associated eigenoperator $\rho_\infty$ is positive definite by assumption. Thus, the matrix form of, for example, $\Phi_{\pi t_0}^*$ in the same basis is:
    \begin{equation*}
        \widehat{\Phi}_{\pi t_0}^* = \begin{pmatrix} D' & 0 \\ 0 & T' \end{pmatrix}
    \end{equation*}
    with $D' = \diag \left( 1, e^{2 \imath \pi^2 / m}, \dots, e^{2 \imath (m-1) \pi^2 / m} \right)$ and $T'$ an upper triangular matrix with diagonal components with modulus strictly less than one. Therefore, $\Phi_{\pi t_0}^*$ is irreducible, since $\rho_\infty$ is its only invariant state. By Theorem 6.6 in \cite{wolf2012quantum}, this implies that the elements in $D'$ remain $m^\mathrm{th}$ roots of unity, which is possible if and only if $m=1$. Consequently, $\Phi_{t_0}^*$ is primitive, and same for all the $\Phi_t^*$, $t>0$.

    \vspace{1mm}

    Let's show that $3 \Leftrightarrow 4$. By construction of $T_{\mathcal{L}^*}$, $\Phi_t^* = e^{t \mathcal{L}^*}$ converges to the projection onto $\Span (\rho_\infty)$. So, for any state $\rho$, $\lim_{t\to\infty} \Phi_t^* (\rho) = \rho_\infty$. On the other hand, item $4$ implies that for any $t>0$,
    \begin{equation*}
        \lim_{k\to\infty} \Phi_{kt}^* (\rho) = \lim_{k\to\infty} \left( \Phi_{t}^* \right)^k (\rho) = \rho_\infty
    \end{equation*}
    which proves item $3$ by Theorem 6.7 in \cite{wolf2012quantum}.

    \vspace{1mm}

    To finish, let's show that $1 \Leftrightarrow 5$. Assuming item $5$, even if some eigenvalues of $\mathcal{L}^*$ are multiples of $2 \imath \pi$, there is a $t_0>0$ such that $\Phi_{t_0}^* = e^{t_0 \mathcal{L}^*}$ has only $\rho_\infty$ as invariant state, so item $1$ is satisfied. On the other hand, if item $5$ is not verified, since each eigenvector in $\ker \left( \mathcal{L}^* \right)$ is a fixed point of $\Phi_t^*$, then none of the $\Phi_t^*$, $t>0$, can be irreducible.
\end{proof}

\vspace{2mm}

This leads to the following Perron-Frobenius-like corollary:
\begin{corollary}
    Assume that $\Phi^*$ satisfies one of the hypotheses of Theorem \ref{Thm_Irred_QMS}. Let $\lambda_2$ be the first eigenvalue of $\mathcal{L}^*$ with negative real part. Then, for any constant $0 < c < - \Re(\lambda_2)$, there exists a constant $C>0$ such that for any state $\rho$:
    \begin{equation*}
        || \Phi_t^* (\rho) - \rho_\infty || \le C e^{- c t}
    \end{equation*}
\end{corollary}
\begin{proof}
    Let $D+N$ be the Jordan normal form of $\mathcal{L}^*$, with $D=\diag \left( 0, \lambda_2, \dots, \lambda_d \right)$. Let $\Phi_\infty^*$ be the projection onto $\Span(\rho_\infty)$. Then, in a certain basis:
    \begin{equation*}
        e^{ct} \left( \widehat{\Phi}_t^* - \widehat{\Phi}_\infty^* \right) = \begin{pmatrix} 0 & & \\ & e^{t(\lambda_2 + c)} & \\ & \ddots & \\ & & e^{t(\lambda_d + c)} \end{pmatrix} e^{tN}
    \end{equation*}
    with 
    $e^{tN} = \sum_{k=1}^d t^k N^k / k!$. So, 
    \begin{equation*}
        \lim_{t\to\infty} e^{ct} \left( \widehat{\Phi}_t^* - \widehat{\Phi}_\infty^* \right) = 0
    \end{equation*}
    Choosing $C$ as the upper bound of the function $t \to e^{ct} ||\Phi_t^* - \Phi_\infty^*||$ finishes the proof.
\end{proof}

\vspace{2mm}

This result, that is very specific to time continuous QMS, can be very powerful since the restriction of the QMS to any minimal enclosure is irreducible by definition.

%%%%%%%%%%%%%%%%%%%%%%%%%%%%%%%%%%%%%%%%%%%%%%%%%%%%%%%%%%%%%%%%%%%%%%%%%%%%%%%%

\section{Conclusion}

In this paper, we studied the qualitative behavior of QMS in infinite-dimensional separable Hilbert spaces, using the positive recurrent subspace $\mathcal{R}_+$. For this, the ergodic theorem of Frigerio and Verri can be a very useful tool.  We provided a sufficient condition that ensures an exponential convergence towards $\mathcal{R}_+$. As application of the ergodic theorem, we considered three examples and showed that they are ergodic QMS. In a second step, we studied more precisely what can happen inside $\mathcal{R}_+$, using the decomposition into minimal enclosures, and we emphasized the fact that this subspace is the most relevant one for studying GAS. Finally, we presented the well-known Perron-Frobenius-like theorem for the special case of irreducible QMS in finite-dimensional Hilbert spaces. 

Finding other criteria to verify the ergodic theorem, in particular using the structure of the fixed points of the QMS, is part of our ongoing research. It may also be possible to define relevant stopping times for some particular QMS, as in classical probability, to characterize attraction of $\mathcal{R}_+$. Another promising approach is to use the theory of parabolic partial differential equations to infer that $A(\mathcal{R}_+) = \mathds{1}$. Finally, generalising the approach used for generic QMS to other types of QMS that also stabilises some abelian subalgebras is very attractive and might be useful for analysing generalisations of the quantum harmonic oscillator.

%%%%%%%%%%%%%%%%%%%%%%%%%%%%%%%%%%%%%%%%%%%%%%%%%%%%%%%%%%%%%%%%%%%%%%%%%%%%%%%%

\bibliographystyle{plain}  
\bibliography{refs}

\end{document}